\documentclass[a4paper,10pt]{article}
\usepackage{test,graphicx,enumerate}

\title{Unification of Maximum Entropy and Bayesian Inference via Plausible Reasoning}
\author{Alexis Akira Toda \thanks{Department of Economics, Yale University.  Email: \href{mailto:alexisakira.toda@yale.edu}{alexisakira.toda@yale.edu}} \thanks{This paper benefited from conversations with Sylvain Barde and Sander Heinsalu.  I am deeply indebted to my thesis advisor, Donald Brown.  I thank Duncan K. Foley for introducing me to the works of Jaynes.  The financial supports from the Cowles Foundation, the Nakajima Foundation, and Yale University are greatly acknowledged.}}
\date{This Version: \today}

\numberwithin{equation}{section}
\newcommand{\maxent}{MaxEnt}
\newcommand{\diff}{\mathrm{d}}

\begin{document}
\maketitle

\begin{abstract}
This paper modifies Jaynes's axioms of plausible reasoning and derives the minimum relative entropy principle, Bayes's rule, as well as maximum likelihood from first principles.  The new axioms, which I call the \emph{Optimum Information Principle}, is applicable whenever the decision maker is given the data and the relevant background information.  These axioms provide an answer to the question ``why maximize entropy when faced with incomplete information?"
\end{abstract}

\section{Introduction}
Bayesian inference \cite{bayes1763} and the maximum entropy principle (\maxent) of Jaynes \cite{jaynes1957a} are valid methods of inference when the decision maker is faced with incomplete information.  Although these methodologies are quite distinct, they often give similar results.  A few authors have hinted at the possibility of deriving both methods from first principles.  For instance, as the sample size increases, \cite{vancampenhout-cover1981} showed that the distribution of a random variable conditional on empirical moment constraints (computed by Bayes's rule) converges to the minimum relative entropy distribution subject to the same population moment constraints.  Conversely, \cite{zellner1988} showed that Bayes's rule can be derived from a variational principle of information processing.

One possibility of deriving \emph{both} the maximum entropy principle and Bayes's rule is to axiomatize plausible reasoning, as \cite{keynes1921,jeffreys1939,cox1946,jaynes2003} attempted.  In the most primitive form, Jaynes \cite{jaynes2003} suggested desiderata that should be employed in plausible reasoning, by which he deduced Bayes's rule.  To apply Bayes's rule we have to start from some priors, and Jaynes advocates the use of the maximum entropy principle to set up priors.  However, there are many situations in which \emph{both} {\maxent} and Bayesian inference are applicable.  Which method should we take then?  And do they return the same result?  In this paper I propose a different set of axioms of plausible reasoning, by which I derive the minimum relative entropy principle\footnote{To the best of my knowledge, the minimum relative entropy principle was first introduced by Kullback \cite[p.~37]{kullback1959} under the name the principle of \emph{minimum discrimination information}.}, Bayes's rule, and maximum likelihood.

I proceed in two steps.  First, I list the desiderata of a measure of information gain when a decision maker updates the plausibility of a proposition upon receiving new information.  From these desiderata I derive the functional form of information gain.  Second, I impose the decision maker to be maximally conservative, given all the relevant information.  That is, the decision maker updates the plausibilities by minimizing the average information gain (\ie, sticks to his or her prior as much as possible) subject to all relevant information, which I call the \emph{Optimum Information Principle}.  I show that the Optimum Information Principle implies the well-known minimum relative entropy principle, the Bayes rule, and also Jaynes's axioms.

\section{Axioms of Plausible Reasoning}
Viewing probability as the plausibility of a proposition dates back at least to Keynes \cite{keynes1921}.  As Cox \cite{cox1946} describes it ``as if Euclid had placed the Pythagorean theorem among the axioms of plane geometry", Keynes's axioms were not fundamental, and have been improved by \cite{jeffreys1939} and \cite{cox1946}.  To date the most primitive axioms of plausible reasoning seem to be those of Jaynes \cite[pp.~17--19]{jaynes2003}:
\begin{enumerate}[J-I.]
\item\label{item:jaynes.1} Degrees of plausibility are represented by real numbers.
\item\label{item:jaynes.2} Qualitative correspondence with common sense.
\item\label{item:jaynes.3} Consistency.
\begin{enumerate}
\item\label{item:jaynes.3a} If a conclusion can be reasoned out in more than one way, then every possible way must lead to the same result.
\item\label{item:jaynes.3b} The robot\footnote{The ``robot" is a machine that performs plausible reasoning according to the desiderata.} always takes into account all of the evidence it has relevant to a question.  It does not arbitrarily ignore some of the information, basing its conclusions only on what remains.  In other words, the robot is completely nonideological.
\item\label{item:jaynes.3c} The robot always represents equivalent states of knowledge by equivalent plausibility assignments.  That is, if in two problems the robot's state of knowledge is the same (except perhaps for the labeling of the propositions), then it must assign the same plausibilities in both.
\end{enumerate}
\end{enumerate}
Desideratum \ref{item:jaynes.2} means the following.  If we denote the plausibility of a proposition $A$ given information $I$ by $p(A|I)$, then
\begin{align}
p(A|C')>p(A|C)&\Longrightarrow p(\lnot A|C')<p(\lnot A|C),~\text{and}\label{eq:2.1}\\
\left.\begin{aligned}p(A|C')&>p(A|C)\\
p(B|A\land C')&=p(B|A\land C)
\end{aligned}\right\}&\Longrightarrow p(A\land B|C')\ge p(A\land B|C).\label{eq:2.2}
\end{align}
In words, \eqref{eq:2.1} says that if information $C$ gets updated to $C'$ in such a way that the plausibility of $A$ is increased, then the plausibility of the negation of $A$ is decreased; \eqref{eq:2.2} says that if, in addition, the plausibility of $B$ given $A$ is unchanged, then the plausibility that both $A$ and $B$ are true must increase.  Chapter 2 of \cite{jaynes2003} shows that desiderata \ref{item:jaynes.1}--\ref{item:jaynes.3b} imply that plausibilities have a probability representation and they obey Bayes's rule, and that desideratum \ref{item:jaynes.3c} implies Laplace's Principle of Indifference \cite{laplace1812} for setting up priors.

In order to derive {\maxent} and Bayes's rule, I first axiomatize the quantity which I call \emph{information gain} and derive its functional form.  The axioms, which are all intuitively appealing, are as follows.
\begin{enumerate}[{IG}-1.]
\item\label{item:e.1} Numerical representation: the information gain $I$ is a function of prior plausibility $p$ and posterior plausibility $q$.
\item\label{item:e.2} Smootheness and monotonicity: the information gain is a smooth, increasing function in posterior plausibility.
\item\label{item:e.3} Path independence: the total information gain of updating the prior plausibility $p$ to the posterior $q$ is independent of the path it is updated.  That is, if there are two paths $p\to r\to q$ and $p\to r'\to q$, then $I(p,r)+I(r,q)=I(p,r')+I(r',q)$.
\item\label{item:e.4} Independence from the choice of unit: whatever unit we choose to describe plausibility, the information gain should have the same value.  That is, $I(tp,tq)=I(p,q)$ for $t>0$.
\item\label{item:e.5} Zero information gain for not updating: for any $p$, we have $I(p,p)=0$.
\end{enumerate}

\begin{prop}\label{prop:e}
Suppose that axioms IG-\ref{item:e.1}--IG-\ref{item:e.5} hold.  Then $I(p,q)=k\log\frac{q}{p}$, where $k>0$ is an arbitrary constant.
\end{prop}
\begin{proof}
Since by axiom IG-\ref{item:e.2} the information gain $I(p,q)$ is smooth in $q$, it is partially differentiable with respect to $q$ and $I$ can be recovered by integrating its partial derivative.  Differentiating $I(p,r)+I(r,q)=I(p,r')+I(r',q)$ with respect to $q$, we get
\begin{equation}
\frac{\partial I}{\partial q}(r,q)=\frac{\partial I}{\partial q}(r',q).\label{eq:e.1}
\end{equation}
The left-hand side of \eqref{eq:e.1} is a function of $(r,q)$, and the right-hand side of \eqref{eq:e.1} is a function of $(r',q)$. Since $r,r'$ are arbitrary, \eqref{eq:e.1} must be a function of only $q$.  Let $\frac{\partial I}{\partial q}(r,q)=g(q)$.  By integration we get $I(r,q)=F(r)+G(q)$, where $F$ is some function and $G=\int g$.  By the path independence axiom IG-\ref{item:e.3}, we get
\begin{align}
[F(p)+G(r)]+[F(r)+G(q)]=&[F(p)+G(r')]+[F(r')+G(q)]\notag\\
\iff &F(r)+G(r)=F(r')+G(r').\label{eq:e.2}
\end{align}
Since \eqref{eq:e.2} holds for any $r,r'$, $F(r)+G(r)$ is constant, but it must be zero by axiom IG-\ref{item:e.5}: $F(r)+G(r)=I(r,r)=0$.  Therefore $I(p,q)=F(p)+G(q)=G(q)-G(p)$.  By axiom IG-\ref{item:e.4}, we have $G(tq)-G(tp)=G(q)-G(p)$.  Differentiating both sides with respect to $q$, we get $tG'(tq)=G'(q)$.  Multiplying both sides by $q$ and letting $x=tq$, we get $xG'(x)=qG'(q)$, so the function $xG'(x)$ is a constant $k$.  Integrating $G'(x)=k/x$ yields $G(x)=k\log x+C$, hence $I(p,q)=G(q)-G(p)=k\log\frac{q}{p}$.  Since $I$ is increasing in $q$ by axiom IG-\ref{item:e.2}, we get $k>0$.  Clearly this function satisfies all axioms IG-\ref{item:e.1}--IG-\ref{item:e.5}.
\end{proof}

From now on let us normalize the arbitrary constant $k$ to 1, so the information gain is given by $I(p,q)=\log\frac{q}{p}$.  This result,
$$\text{information gain}=\log\frac{\text{posterior plausibility}}{\text{prior plausibility}},$$
is mathematically identical to \cite[p.~4]{goldman1953}, although Goldman takes this as the definition.\footnote{In information theory the quantity $-\log p$ is known as the \emph{self-information}, although I was unable to find a reference for its origin (Tribus \cite{tribus1961} calls it \emph{surprisal}).  Our information gain $I(p,q)=\log\frac{q}{p}$ is the difference of the self-information of the prior and posterior.  Kullback and Leibler \cite{kullback-leibler1951} call $\log\frac{p_1}{p_2}$ the \emph{information for discrimination}, where $p_1,p_2$ are general probabilities and not necessarily the prior and the posterior.  The prior/posterior interpretation of $p$ and $q$ can also be clearly seen in \cite{hobson1969,hobson-cheng1973}.}

In order to make plausible reasoning based on available information, consider the following desiderata.

\begin{enumerate}[I.]
\item\label{item:reason.1} Degrees of plausibility are represented by probabilities.
\item\label{item:reason.2} The robot always takes into account all of the evidence it has relevant to a question.  It does not arbitrarily ignore some of the information, basing its conclusions only on what remains.  In other words, the robot is completely nonideological.
\item\label{item:reason.3} Aristotelian logic: the robot assigns zero plausibility to propositions that contradict its knowledge.
\item\label{item:reason.4} The robot always represents equivalent states of knowledge by equivalent plausibility assignments.  That is, if in two problems the robot's state of knowledge is the same (except perhaps for the labeling of the propositions), then it must assign the same plausibilities in both.
\item\label{item:reason.5} Given prior plausibilities, the robot updates the plausibilities by minimizing the average information gain of the posterior plausibilities subject to known information.  In other words, the robot is maximally conservative.
\end{enumerate}

Desideratum \ref{item:reason.1} is stronger than Jaynes's desideratum J-\ref{item:jaynes.1} because I assume that the plausibility is a probability (\ie, finitely or countably additive measure).  In particular, the plausibilities of mutually exclusive propositions are additive: if $A,B$ are mutually exclusive propositions, then $p(A\lor B)=p(A)+p(B)$.  Desideratum \ref{item:reason.2} is identical to J-\ref{item:jaynes.3b}.  Desideratum \ref{item:reason.3} might be interpreted as a special case of \ref{item:reason.2} and probably needs no justification, but I need it nevertheless.  Desideratum \ref{item:reason.4} is identical to J-\ref{item:jaynes.3c}, Laplace's Principle of Indifference, which may or may not be necessary to prove subsequent theorems.

Desideratum \ref{item:reason.5} is the major difference from Jaynes's axioms.  While Jaynes imposes ``qualitative correspondence with common sense" (J-\ref{item:jaynes.2}), I impose that the robot is maximally conservative. This axiom makes sense, for if the robot radically updates the plausibilities (\ie, not sticking to its prior), then it should not have set up the particular prior plausibilities in the first place.  To avoid unnecessary reference to axiom numbers, let us group the desiderata as follows:
\begin{table}[htbp]
\centering
\begin{tabular}{ll}
\ref{item:reason.1}--\ref{item:reason.3}:& Weak Axioms of Plausible Reasoning\\
\ref{item:reason.1}--\ref{item:reason.4}:& Strong Axioms of Plausible Reasoning\\
IG-\ref{item:e.1}--IG-\ref{item:e.5} and \ref{item:reason.5}:&Minimum Information Gain Principle
\end{tabular}
\end{table}

\section{Implications of the Axioms}
In this section I show that the new axioms imply Bayesian inference, maximum likelihood, maximum entropy principle, and minimum relative entropy principle.
\begin{thm}\label{thm:MREP}
Weak plausibility and the minimum information gain principle imply the minimum relative entropy principle (the minimum discrimination information principle of Kullback \cite[p.~37]{kullback1959}).
\end{thm}
\begin{proof}
Let $\set{A_i}$ be propositions that are mutually exclusive and exhaustive.  Let $p_i=p(A_i|I)$ be the prior plausibility of proposition $A_i$ given background information $I$, and $q_i=p(A_i|I')$ be the posterior plausibility to be computed given the new information $I'$.  By desideratum \ref{item:reason.1}, we have $p_i,q_i\ge 0$ and $\sum p_i=\sum q_i=1$.  Since by Proposition \ref{prop:e} the information gain of $A_i$ is $\log\frac{q_i}{p_i}$, the \emph{ex post} average information gain is
$$H(q;p):=\sum_{i=1}^nq_i\log \frac{q_i}{p_i},$$
the relative entropy.\footnote{This quantity was first proposed by Kullback and Leibler \cite{kullback-leibler1951}, which they call, appropriately, ``information".}  By desiderata \ref{item:reason.2} and \ref{item:reason.5}, the robot minimizes $H(q;p)$ subject to all known information $I'$ and the constraints $q_i\ge 0$, $\sum q_i=1$, which is precisely the minimum relative entropy principle.
\end{proof}

\begin{cor}\label{cor:MEP}
Strong plausibility and the minimum information gain principle imply the maximum entropy principle of Jaynes \cite{jaynes1957a} for setting up priors.
\end{cor}
\begin{proof}
Desideratum \ref{item:reason.4} is nothing but Laplace's Principle of Indifference.  Hence, by desideratum \ref{item:reason.1}, the robot assigns the prior plausibility $p(A_i)=\frac{1}{n}$.  By Theorem \ref{thm:MREP} the robot computes the posterior plausibility $p_i=p(A_i|I)$ by minimizing
$$\sum_{i=1}^n p_i\log \frac{p_i}{1/n}=\sum_{i=1}^n p_i(\log p_i+\log n)=\sum_{i=1}^n p_i\log p_i+\log n,$$
(where we have invoked desideratum \ref{item:reason.1}: $\sum p_i=1$) or equivalently, maximizing Shannon's entropy $H(p)=-\sum_{i=1}^n p_i\log p_i$ \cite{shannon1948}.  This is precisely Jaynes's maximum entropy principle \cite{jaynes1957a}.
\end{proof}

I propose to define the Optimum Information Principle by the combination of the weak or strong plausibility and the minimum information gain principle, despite its implication is the well-known minimum relative entropy principle.  There are two reasons to avoid the term ``entropy". First, ``entropy" is a misnomer both in physics (see \cite{ben-naim2008}) and in information theory.  According to \cite{tribus-mcirvine1971}, Shannon \cite{shannon1948} named his measure of uncertainty or missing information ``entropy" following the advice of von Neumann: ``[It] has been used in statistical mechanics under that name \dots [and] no one knows what entropy really is, so in a debate you will always have the advantage."  Clausius coined the word ``entropy" after the Greek word for ``transformation"; given that ``entropy" is a misnomer, adding the adjective ``relative" makes it only worse.  Second, as a measure of information gain the Kullback-Leibler information $H(q;p)$ is more fundamental than the Shannon entropy $H(p)$ as shown by the above axiomatic derivation as well as the comparison of the two information measures provided in \cite{hobson-cheng1973}: the Kullback-Leibler information, unlike the Shannon entropy, extends to arbitrary probability measures and it satisfies an additivity property.  Since by desideratum \ref{item:reason.5} the quantity $H(q;p)$ is the average information gain, and since Kullback and Leibler \cite{kullback-leibler1951} call $H(q;p)$ ``information" before the term ``relative entropy" was coined, the term Optimum Information Principle seems best.\footnote{\cite{haken2004} calls it the Maximum Information Principle, meaning that the missing information is maximized.  Maximizing the missing information is equivalent to minimizing the information gain as we do here.  The adjective ``optimum" avoids the confusion between maximum/minimum.}

\begin{thm}\label{thm:Bayes}
Weak plausibility and the minimum information gain principle imply Jaynes's desiderata \ref{item:jaynes.1}--\ref{item:jaynes.3b}, in particular Bayes's rule.  Therefore, the Optimum Information Principle is consistent with Bayesian inference.
\end{thm}
\begin{proof}
Let us first prove Bayes's rule.  Suppose that the robot is given background information $I$ and that the robot has prior plausibilities on the propositions $A_1,\dotsc,A_n$, $B$, and any logical conjunction or negation generated by them.  Therefore the prior plausibilities of $A_i\land A_j$, $A_i\land B$, $A_i\land(\lnot B)$, etc., which are denoted by $p(A_i\cap A_j|I)$, $p(A_i\cap B|I), p(A_i\cap B^c|I)$, etc., are well defined.  The task of the robot is to update the plausibilities of $\set{A_i}$ when it is given additional information $B$.  Since there are only a finite number of propositions, without loss of generality we may assume that $\set{A_i}$ are mutually exclusive and exhaustive.  By desideratum \ref{item:reason.1}, we have $\sum_{i=1}^np(A_i|I)=1$.

Let us denote the posterior plausibilities by $q(A_i\cap B|B\cap I)$, etc.  In order to compute them, by Theorem \ref{thm:MREP} the robot solves
\begin{subequations}\label{eq:2.3}
\begin{align}
&\min_q \sum q\log\frac{q}{p}\quad \text{subject to}\label{eq:2.3a}\\
&\sum_{i=1}^n (q(A_i\cap B|B\cap I)+q(A_i\cap B^c|B\cap I))=1,\label{eq:2.3b}\\
&(\forall i)~q(A_i\cap B^c|B\cap I)=0,\label{eq:2.3c}
\end{align}
\end{subequations}
where $p,q$ in \eqref{eq:2.3a} take all possible forms of $p(A_i\cap B|I), q(A_i\cap B|B\cap I)$ and $p(A_i\cap B^c|I), q(A_i\cap B^c|B\cap I)$.  Conditions \eqref{eq:2.3b} and \eqref{eq:2.3c} come from desiderata \ref{item:reason.1} and \ref{item:reason.3}: since $\lnot B$ (and hence $A_i\land (\lnot B)$) is logically impossible knowing $B$, the robot assigns zero plausibility to $A_i\land (\lnot B)$.  That we impose \eqref{eq:2.3b} and \eqref{eq:2.3c} and nothing else comes from using all relevant information as in desideratum \ref{item:reason.2}.

Since the function $f(q)=q\log \frac{q}{p}$ is continuous and strictly convex and the constraints \eqref{eq:2.3b}, \eqref{eq:2.3c} constitute a compact convex set, we can apply the Karush-Kuhn-Tucker theorem to solve \eqref{eq:2.3}.  Let $\lambda$ be the Lagrange multiplier corresponding to \eqref{eq:2.3b}.  The Lagrangian is
$$L(q,\lambda)=\sum_{i=1}^nq_i\log\frac{q_i}{p_i}+\lambda\left(\sum_{i=1}^nq_i-1\right),$$
where we have used the shorthand $q_i=q(A_i\cap B|B\cap I)$ and $p_i=p(A_i\cap B|I)$.  The first-order condition, which is necessary and sufficient, reads $$\frac{\partial L}{\partial q_i}=\log \frac{q_i}{p_i}+1+\lambda=0.$$
This shows that $q_i$ is proportional to $p_i$, so by $\sum q_i=1$ we obtain $q_i=p_i/\sum_{i=1}^n p_i$.  Therefore,
\begin{align}
q(A_i|B\cap I)&=q(A_i\cap B|B\cap I)+q(A_i\cap B^c|B\cap I)\notag\\
&=q(A_i\cap B|B\cap I)=q_i\notag\\
&=\frac{p(A_i\cap B|I)}{\sum_{i=1}^n p(A_i\cap B|I)}=\frac{p(A_i\cap B|I)}{p(B|I)},\label{eq:2.4}
\end{align}
where the first equality holds because $q$ is a probability (desideratum \ref{item:reason.1}), the second equality holds because $q(A_i\cap B^c|B\cap I)=0$ (desideratum \ref{item:reason.3}), and the last equality holds because $p$ is a probability and $\set{A_i}$ are mutually exclusive and exhaustive.  \eqref{eq:2.4} is precisely the Bayes rule.

Now let us show that Jaynes's desiderata \ref{item:jaynes.1}--\ref{item:jaynes.3b} are implied.  All we need to show are desiderata \ref{item:jaynes.2} (conditions \eqref{eq:2.1} and \eqref{eq:2.2}) and \ref{item:jaynes.3a}.  \eqref{eq:2.1} holds because plausibility has a probability representation by desideratum \ref{item:reason.1}.  \eqref{eq:2.2} holds by Bayes's rule, which we have already deduced in \eqref{eq:2.4}.  Desideratum \ref{item:jaynes.3a} holds by the Additivity Theorem of Hobson and Cheng \cite[p.~308]{hobson-cheng1973}, where they essentially show that if the robot has initial background information $I_0$ that gets updated to $I_1$ and then to $I_2$, with plausibilities $p_0,p_1,p_2$ respectively, then
$$H(p_2;p_0)=H(p_2;p_1)+H(p_1;p_0),$$
that is, the (minimized) Kullback-Leibler information is additive.\footnote{This property is mathematically equivalent to the ``Subset Independence" axiom of Shore and Johnson \cite[p.~27]{shore-johnson1980}.}  In particular, if there are two ways to update, $I_0\to I_1\to I_2$ and $I_0\to I_1'\to I_2$, then we obtain
\begin{equation}
H(p_2;p_1)+H(p_1;p_0)=H(p_2;p_1')+H(p_1';p_0),\label{eq:2.5}
\end{equation}
the path independence.  Therefore Jaynes's desideratum \ref{item:jaynes.3a} holds because if a conclusion can be reasoned out in more than one way, the path independence property \eqref{eq:2.5} ensures that every possible way leads to the same result.
\end{proof}

At this point I stress the distinction between our axiomatization and other author's.  In his seminal work \cite{shannon1948}, Shannon imposes as the third axiom ``If a choice be broken down into two successive choices, the original $H$ should be the weighted sum of the individual values of $H$", in which he implicitly uses Bayes's rule.  The same remark applies to the axiomatization of the Kullback-Leibler information by Hobson \cite{hobson1969}.  Similarly, in the important axiomatization of the maximum entropy principle, Shore and Johnson \cite{shore-johnson1980} implicitly use Bayes's rule in their fourth axiom ``Subset Independence: It should not matter whether one treats an independent subset of system states in terms of a separate conditional density or in terms of the full system density".  Zellner \cite{zellner1988} derives Bayes's rule from an ``information processing rule", but it is not clear how it relates to maximum entropy and his definition of information seems somewhat arbitrary.  On the contrary, Cox \cite{cox1946} and Jaynes \cite{jaynes2003} derive the Bayes rule from intuitively appealing first principles, as we have done.  In addition we have derived the maximum entropy principle and the minimum relative entropy principle.

Finally, let us show that the Optimum Information Principle implies maximum likelihood.
\begin{thm}\label{thm:ML}
The Optimum Information Principle implies the maximum likelihood principle of Fisher \cite{fisher1912}.\footnote{Although maximum likelihood is attributed to Fisher, it was already used by Laplace and Gauss a century before.}
\end{thm}

\begin{proof}
Suppose that $\set{X_n}_{n=1}^N$ are independently and identically distributed random variables with an unknown density $f$.  Given the realizations $\set{x_n}$, suppose that the statistician wishes to fit a parametric density $f(x;\theta)$ to $f$, where $\theta\in\Theta$ is a parameter.  Although prior and posterior distributions are meaningless for a frequentist, it is natural to interpret that the model $f(x;\theta)$ and the truth $f$ correspond to the prior and posterior, respectively.  Hence to make an optimal inference the statistician should choose $\theta$ so as to minimize the Kullback-Leibler information
$$H(f;f_\theta)=\int f(x)\log\frac{f(x)}{f(x;\theta)}\diff x.$$
However, by the law of large numbers we obtain
\begin{align*}
H(f;f_\theta)&=\int f\log f-\int f\log f_\theta=\int f\log f-\E_f[\log f(X;\theta)]\\
&\approx \int f\log f-\frac{1}{N}\sum_{n=1}^N\log f(x_n;\theta),
\end{align*}
so the statistician should maximize the log likelihood $\sum_{n=1}^N\log f(x_n;\theta)$.
\end{proof}

\section{Concluding Remarks}
The maximum entropy principle has occasionally been criticized \emph{ad hoc} as ``Why maximize entropy (or minimize relative entropy), why not other functions?".  An inference method is valuable if and only if it is useful in analyzing real data, and therefore an inference method requires no interpretation, and no justification except practical usefulness. (Nevertheless the justification of the maximum entropy principle has been provided \cite{jaynes1957a,jaynes1968,jaynes1982,shore-johnson1980}.)  It is well-known that the minimum relative entropy principle (maximum entropy principle) and Bayesian inference are useful (see \cite{loredo1990,jaynes2003} and the references therein).  Therefore, since our Optimum Information Principle implies the minimum relative entropy principle, Bayes's rule, as well as maximum likelihood, it should be equally useful.  In addition we have axiomatized plausible reasoning and derived the maximum entropy principle; hence we have answered the question ``why maximize entropy?"

\bibliographystyle{plain}


\end{document}